\documentclass[10pt, conference,english, letterpaper,onecolumn]{IEEEtran}
\IEEEoverridecommandlockouts

\usepackage[cmex10]{amsmath}
\usepackage[noadjust]{cite}
\usepackage[mathscr]{euscript}
\usepackage[dvipsnames]{xcolor}
\usepackage[T1]{fontenc}
\usepackage{%
	amsfonts,%
	amssymb,%
	amsthm,%
	bbm,%
	enumerate,%
	float,%
	dblfloatfix,%
	dsfont,%
	mathtools,%
	url,%
	pgf,%
	pgfplots,%
	tikz,%
	subcaption,%
}
\usepackage{graphicx} 
\graphicspath{ {./Figures/} }

\usepgflibrary{shapes}
\usetikzlibrary{%
arrows,%
backgrounds,%
calc,%
calendar,
chains,%
decorations,
decorations.pathmorphing,
fit,%
matrix,%
mindmap,
petri,%
positioning,
scopes,%
shadings,%
shadows,%
shapes,%
shapes.arrows,%
shapes.misc,
shapes.symbols,%
}

\theoremstyle{plain}
\newtheorem{thm}{Theorem}

\newtheorem{prop}[thm]{Proposition}

\newtheorem{appr}{Approximation}

\theoremstyle{definition}

\newtheorem{prob}{Problem}

\theoremstyle{remark}
\newtheorem{rem}{Remark}

\newcommand{\EQ}[1]{\begin{displaymath}#1\end{displaymath}}
\newcommand{\EQN}[1]{\begin{equation}#1\end{equation}}
\newcommand{\eq}[1]{\begin{align*}#1\end{align*}}

\newcommand{\abs}[1]{\lVert#1\rVert}
\newcommand{\set}[1]{\left\{#1\right\}}
\newcommand{\SetIn}[1]{\mathds{1}_{\set{#1}}}
\newcommand{\ceil}[1]{\lceil#1\rceil}
\newcommand{\floor}[1]{\lfloor#1\rfloor}

\newcommand{\iid}{\emph{i.i.d.}\ }
\renewcommand{\le}{\leqslant}
\renewcommand{\ge}{\geqslant}

\newcommand{\bbC}{\mathbb{C}}
\newcommand{\E}{\mathbb{E}}

\newcommand{\R}{\mathbb{R}}
\newcommand{\Z}{\mathbb{Z}}

\newcommand{\sZ}{\mathscr{Z}}

\newcommand\copyrighttext{%
	\footnotesize \textcopyright 2021 IEEE. Personal use of this material is permitted. Permission from IEEE must be obtained for all other uses, in any current or future media, including reprinting/republishing this material for advertising or promotional purposes, creating new collective works, for resale or redistribution to servers or lists, or reuse of any copyrighted component of this work in other works. }
\newcommand\copyrightnotice{%
	\begin{tikzpicture}[remember picture,overlay]
	\node[anchor=south,yshift=10pt] at (current page.south) {\fbox{\parbox{\dimexpr\textwidth-\fboxsep-\fboxrule\relax}{\copyrighttext}}};
	\end{tikzpicture}%
}

\title{Latency-Redundancy Tradeoff in Distributed Read-Write Systems}
\author{
Saraswathy Ramanathan\IEEEauthorrefmark{1}%
\and Gaurav Gautam\IEEEauthorrefmark{1}%
\and Vikram Srinivasan\IEEEauthorrefmark{1}%
\and Parimal Parag\IEEEauthorrefmark{1}
\thanks{
The authors\IEEEauthorrefmark{1} are with Indian Institute of Science, Bangalore, KA 560012, India. 
Email: \IEEEauthorrefmark{1}\{saraswathyr,gauravgautam,vikramsriniv,parimal\}@iisc.ac.in.
}
}

\begin{document}
\maketitle
\copyrightnotice
\begin{abstract}
Data is replicated and stored redundantly over multiple servers for availability in distributed databases.  
We focus on databases with frequent reads and writes, where both read and write latencies are important. 
This is in contrast to databases designed primarily for either read or write applications. 
Redundancy has contrasting effects on read and write latency. 
Read latency can be reduced by potential parallel access from multiple servers, 
whereas write latency increases as a larger number of replicas have to be updated. 
We quantify this tradeoff between read and write latency as a function of redundancy, 
and provide a closed-form approximation when the request arrival is Poisson and the service is memoryless. 
We empirically show that this approximation is tight across all ranges of system parameters. 
Thus, we provide guidelines for redundancy selection in distributed databases.
\end{abstract}


\section{Introduction}
Storage systems are designed with specific applications in mind. 
In this article, we focus on the systems where read and write are both frequent, 
and we would refer to these as \emph{read-write systems}. 
Some examples of cloud systems with frequent reads and writes are banking, personal storage, e-commerce, etc.
Cloud storage systems like Dropbox, GitHub, OneDrive, Google Drive etc. have frequent updates (writes) to the same file and benefits from study of systems with frequent reads and writes.
Personal storage and banking receive read and write requests of the same order. 
In a personal storage cloud like Dropbox, the daily average of uploaded files is $1.2$ billion. 
Dropbox receives $1.67$ billion API calls in a day, of which $345.6$ million ($\sim21\%$) are edits to files~\footnote{\url{https://expandedramblings.com/index.php/Dropbox-statistics/}}. 
State Bank of India receives around $131.16$ million transactions per month\footnote{\url{https://www.business-standard.com/company/st-bk-of-india-1375/annual-report/director-report}} from $296.82$ million page visits\footnote{\url{https://www.similarweb.com/website/retail.onlinesbi.com/#overview}}. 
Out of the total number of queries sent to the bank cloud servers, 
$44\%$ are write requests.
We will focus on the latency performance of these systems, 
which is an important user requirement that has monetary impact. 


Distributed read-write systems are employed in many modern storage and computing architectures, 
for graceful scaling up. 
There are several important considerations in the design and implementation of such distributed systems, 
such as consistency, latency, availability, storage cost, among others. 
Availability in the event of failures, is ensured by redundant storage of data over multiple servers, in these systems. 

Most commercial distributed database systems provision for eventual consistency~\cite{DeCandiaSIGOPS2007,BailisMQ2013}, 
where read requests can access an older version of data. 
We consider the read and write latencies for eventually consistent systems. 
In addition, we adopt the primary-secondary architecture with redundant replication for distributed read-write systems as shown in Fig.~\ref{fig:MySQL}. 
This architecture is employed by popular databases such as MySQL, DynamoDB, MongoDB, PostgreSQL, etc.  
As shown in Fig.~\ref{fig:MySQL}, write requests arrive at the designated \emph{primary} server, 
and the remaining \emph{secondary} servers copy the written data from the primary.   
In contrast, multiple read requests to a file can be served simultaneously. 
Data read requests can be directed to any server in the cluster holding a copy of the data.
Often read and write requests are stored in separate queues, 
and depending on the application, one of them is prioritized over the other. 
We consider distributed read-write systems for the two different priority instances: read and write priority. 

\begin{figure}[h]
\centering
\scalebox{1.0}{\begin{tikzpicture}
[font=\fontsize{6.75pt}{7.5}\selectfont, line width=0.75pt, node distance=0mm, draw=black, >=stealth',
server/.style={cylinder, shape border rotate=90, aspect=0.35, minimum height=9mm, minimum width=10.5mm, draw=black, inner sep=0pt},
client/.style={rectangle, minimum height=6mm, minimum width=11.25mm, draw=black, inner sep=0pt},
router/.style={circle, minimum size=11mm, draw=black, fill=gray!40, text width=3.25em, text centered, inner sep=0pt},
system/.style={rectangle, rounded corners=1.5mm, minimum height=37.5mm, minimum width=62.25mm, draw=black, dashed}
]

\node[server](M) at (0,15mm) {Master};
\node[server](S1) at (-12.75mm,0) {Slave1};
\node[server](S2) at (0,0) {Slave2};
\node[server](S3) at (12.75mm,0) {Slave3};
\node[router](R) at (-30mm,7.5mm) {Dispatcher};
\node[system](MySQL) at (-8.625mm,5.25mm) {};

\node[client](C1) at (-52.5mm,15mm) {Clients};
\node[client](C2) at (-52.5mm,7.5mm) {Clients};
\node[client](C3) at (-52.5mm,0) {Clients};

\draw[->, color=black!50!red, densely dashdotted] ([xshift=-0.375mm]R.north) -- ([xshift=-24.9375mm,yshift=0.75mm]M.west) -- 
node[above, color=black!50!red] {Write} ([yshift=0.75mm]M.west);
\draw[->, color=black!50!red, densely dashdotted] (M.south) -- (S2.north);
\draw[->, color=black!50!red, densely dashdotted] ([yshift=3.75mm]S2.north) -- ([yshift=3.75mm]S1.north) -- (S1.north);
\draw[->, color=black!50!red, densely dashdotted] ([yshift=3.75mm]S2.north) -- ([yshift=3.75mm]S3.north) node[above, color=black!50!red] {Replication} -- (S3.north);

\draw[->, color=black!60!green] ([xshift=0.375mm]R.north) -- ([xshift=-24.1875mm]M.west) -- node[below, color=black!60!green] {Read} (M.west);
\draw[->, color=black!60!green] ([xshift=-0.75mm]R.south) -- ([xshift=-0.75mm,yshift=-10.5mm]R.south) -- 
node[below, color=black!60!green] {Read} ([yshift=-4.5mm]S3.south) -- (S3.south);
\draw[->, color=black!60!green] (R.south) -- ([yshift=-9.75mm]R.south) -- ([yshift=-3.75mm]S2.south) -- (S2.south);
\draw[->, color=black!60!green] ([xshift=0.75mm]R.south) -- ([xshift=0.75mm,yshift=-9mm]R.south) -- ([yshift=-3mm]S1.south) -- (S1.south);

\draw[->]  (C1.east) -- (R);
\draw[->]  (C2.east) -- node[left = 1.5mm, above, black] {R/W} (R);
\draw[->]  (C3.east) -- (R);

\end{tikzpicture}}
\caption{Distributed read-write system with primary-secondary architecture.}
\label{fig:MySQL}
\end{figure}
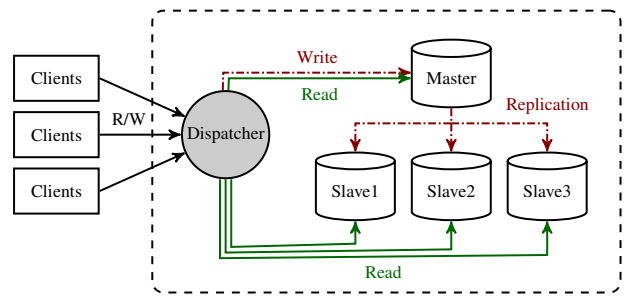

In this work, we are interested in I/O bound distributed database systems with geographically co-located data, 
where all the requests suffer from a similar negligible network delay. 
Accordingly, our focus is on the queueing delay suffered by the read and write requests. 
This is the scenario for many SaaS companies hosted on cloud service providers, 
with the database situated in one or two availability zones. 
Our analysis framework is also suitable for the cases when network latency does not scale with redundancy, 
and can be accounted for by an additional network latency. 
Redundant storage of data is advantageous for read latency,  
as it allows for parallel access. 
It can be shown\footnote{\url{https://docs.gitlab.com/ee/administration/database_load_balancing.html}}
that in many practical situations, 
the read latency decreases with increase in redundancy. 
Contrastingly, the write latency increases with redundancy~\cite{AbadiComp2012, ZhongINFOCOM2018}, 
as a write should be completed at all redundant copies of the data. 
This alludes to a tradeoff between read and write latency with increased redundancy, 
as illustrated in Fig.~\ref{fig:Tradeoff}, 
where we observe that there exists an optimal redundancy that minimizes the average request latency (averaged over all read and write requests) for a single file.
A quantitative characterization of read and write latency is the main objective of this article. 
We note that, the read and write latencies can be weighted depending on the application. 
For simplicity, we consider the case when they are equally weighted. 
\begin{figure}[h]
\centering
\scalebox{1.0}{\begin{tikzpicture}[scale=0.65]
\pgfplotsset{every axis ylabel/.append style={font=\large},
	xlabel/.append style={font=\large}}
\pgfplotsset{every tick label/.append style={font=\normalsize}}

\begin{axis}[
  xlabel={Number of servers},
  ylabel={Mean request latency (in s)},
  xmajorgrids,
  ymajorgrids,
  xmin=1,
  xmax=7,
  legend style = {legend pos = north east, nodes=right, font=\large},
]


\addplot[
color=blue,
line width=1.2pt,
every mark/.append style={solid},
]
coordinates {
(1,19.623125) (2,17.713125) (3,17.716875) (4,17.953125) (5,18.164375) (6,18.37875) (7,18.525)
};

\end{axis}
\end{tikzpicture}}
\caption{Empirical mean latency of requests in a distributed read-write system, with increasing number of servers.}
\label{fig:Tradeoff}
\end{figure}
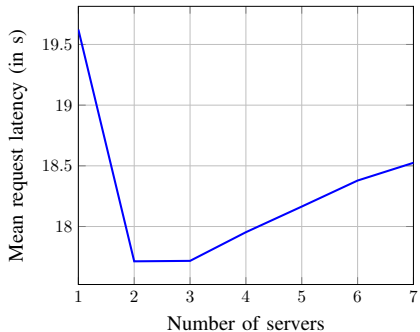

\subsection{Related Work} 
Redundancy schemes and request scheduling for read latency reduction in distributed storage systems has been studied in~\cite{HuangISIT2012, ShahArXiv2012, JoshiAllerton2012, ChenINFOCOM2014, LiangTNET2014, ShahTCOM2016, XiangTNET2016, ParagINFOCOM2017, Al-AbbasiINFOCOM2018}. 
Latency for a single class of requests with coded redundancy in distributed storage and compute systems has been studied in~\cite{HuangISIT2012, JoshiJSAC2014,GardnerSIGMETRICS2015,ShahTCOM2016,XiangTNET2016,Al-AbbasiTNET2018, Al-AbbasiTNSM2019}. 
Incoming requests are forked to all redundant servers, 
and a request is considered completed when a subset of servers finish completion. 
These systems are called \emph{fork-join} queues and have been studied for memoryless arrivals and service in~\cite{JoshiAllerton2012, JoshiJSAC2014, BaditaTIT2019, GardnerQUESTA2016, WangQUESTA2019}. 
For an eventually consistent read-write system with replication redundancy and instantaneous reads, 
staleness of reads is characterized in~\cite{ZhongINFOCOM2018}. 
All the above mentioned works focus on a single class of requests (either read or write), 
assuming either an instantaneous service for the other request class or immutability of data.   
In practical storage systems, however, the data is written and read with non-negligible workload from both. 
We focus on distributed systems where the servers are deployed within the same availability zone and thus have similar network latency for all requests. 
Thus, the network latency does not scale with redundancy and can be ignored, unlike~\cite{UluyolNSDI2020} that study latency in geo-distributed systems.

\subsection{Main contribution}

We analytically compute the read and write latency for a single file in a distributed read-write system, 
and obtain the optimal redundancy to minimize the aggregate latency.

\begin{enumerate}[(i)]
\item Unlike previous works on latency reduction, 
we are considering both read and write latency, without assuming one of them to be instantaneous.

\item We closely model real world distributed storage systems, 
allowing the read and write requests to have different queueing policies. 
In addition, we consider prioritized reads and writes.  

\item We characterize the read and write processes separately, 
while accounting for their joint presence.  
This allows the optimization problem to be tailored for different applications based on read and write latency constraints. 

\item Further, we provide approximations with closed-form expressions for latency redundancy trade-off, 
that can be used for large-scale system design. 
We remark that the Markov chain is more complex than previously studied fork-join queues, 
for read priority systems. 
These queues have not been analytically studied in the literature to the best of our knowledge. 
We empirically show that proposed approximations remain tight over the entire range of system parameters in the system stability region.  

\item As a consequence of our analysis, we show that the optimal number of servers depends on the traffic pattern in the system. 
Hence, from the system design perspective, it is not always beneficial to set redundancy factor to the typical value of two. 

\item We conducted numerical experiments for read-write systems with non-memoryless service distributions. 
In addition, we performed empirical experiments on real world read-write systems. 
We observed the existence of optimal redundancy in both of these situations, 
which confirms that the insights obtained from our theoretical studies continue to hold in general.
\end{enumerate}

\section{System model} 
We consider a distributed read-write system with a primary-secondary architecture. 
In such systems, write occurs at the primary first and then replicated at the secondary servers, whereas reads can occur from any server. 
For simplicity, we focus on a single file stored at the primary and $n$ secondary servers. 
However, the framework can be extended to study systems with multiple files as well. 

\subsection{Arrivals of read/write requests}
We assume that the read and the write requests for the file arrive as Poisson processes with rates $\lambda_r$ and $\lambda_w$ respectively. 
This is a widely accepted model for arrivals in distributed storage~\cite{JoshiJSAC2014,XiangTNET2016,LiPECS2018}
and caching systems~\cite{LiINFOCOM2018,QuanINFOCOM2019,ZhangArXiv2020}.
This assumption is motivated by analytical tractability, and the fact that this is a good approximation for the arrivals~\cite{IversenDTU2015}
when a large number of independent clients are reading from and writing to the system.
In the following subsections, we discuss the modeling assumption on read and write processes separately. 

\subsection{Dispatch of read/write requests} 
We can distinguish read and write request arrivals as two separate classes of arrivals. 
Note that all $(n+1)$ servers receive both read and write requests. 

\subsubsection{Read requests} 
We assume that incoming read requests are dispatched to one of the $(n+1)$ servers uniformly at random, independent of all other decisions.  
This is thinning of the Poisson process~\cite{ChiuWiley2013},
and hence the read request arrival at each server is a Poisson process with rate $\frac{\lambda_r}{(n+1)}$. 
We note that in a typical distributed read-write system with primary-secondary architecture, 
incoming read requests are directed to one of the $n+1$ servers in a round-robin fashion~\cite{MySQLDocumentation}.
We remark that even though optimal routing scheme would be to join the shortest queue or a variant, 
these schemes have communication overhead. 
Therefore, many practical systems employ round robin scheduling for simplicity.
This results in an effective arrival rate of ${\lambda_r}/{(n+1)}$ at each server, 
identical to the Poisson arrival rate achieved by the random splitting of Poisson process. 
Since the two arrival process only differ in higher order moments, we assume the random splitting 
for analytical tractability. 

\subsubsection{Write requests} 
In a primary-secondary architecture, the write request joins the write queue at the primary. 
After the write is completed at the primary, 
the request is forked to all $n$ secondary servers. 
A write is considered completed, if write request is completed at all $n$ secondary servers.

\subsection{Scheduling}  
We assume there is a priority order between read and write classes, 
and requests within a class are served in a first come first serve (FCFS) manner at each server. 

\subsubsection{Priority between classes}
In many distributed read-write systems, one class has priority over the other\footnote{\url{https://mariadb.com/kb/en/high_priority-and-low_priority/}}~\cite{MySQLDocumentation}.
In practical systems, priorities are non-preemptive~\cite{ShankerIGI2020}.
That is, if a request of higher priority class arrives when the server is serving a request of low priority class, 
then the higher priority request has to wait until the request in service is completed. 
For simplicity, we consider preemptive resume priority~\cite{PhillipsMPMOS1998, AtarAAP2001},
where a high priority arrival replaces a low priority request from service, 
and the low priority request resumes its service once there are no more high priority requests. 
This reduces the state space, and makes analysis tractable. 
We will treat the read and write priority systems separately, 
as the system evolution is different for both.

\subsubsection{FCFS within a class}
In practical systems, write requests are served in an FCFS manner~\cite{FraserThesis2004},
while the read requests are served by processor sharing\footnote{\url{https://www.sqlshack.com/locking-sql-server/}}~\cite{MySQLDocumentation}.
According to the insensitivity property~\cite{GiambeneSpringer2014},
all work conserving policies that do not depend on service time of requests have the same mean waiting time.
Since we are only interested in the expected behavior of the system for our analysis, 
we assume that both read and write requests are served in FCFS manner within their class. 

\subsection{Execution of read/write requests} 

The uncertainty in the execution time of the request at each server occurs due to independent background processes at individual servers, 
and hence the execution time can be modeled by independent random variables, both across the servers and the requests.  
We also assume that homogeneous servers such that each execution time has identical distribution, 
depending on the request class.  

\subsubsection{Read requests} 

Real traces from Amazon S3 show that the empirical distribution of time to retrieve a fixed size chunk of data can be well approximated by an exponential distribution~\cite{ChenINFOCOM2014}. 
This assumption makes the analysis tractable, 
and is a popular assumption for content download time in literature~\cite{HuangArXiv2012, HuangISIT2012, JoshiAllerton2012, JoshiJSAC2014, BaditaTIT2019}.
Accordingly, we assume that the read times are \iid across read requests and servers, 
distributed exponentially with rate $\mu_r$. 

\subsubsection{Write requests} 
It has been shown that the random write latency in distributed storage systems, 
can be well modeled by shifted exponential distribution~\cite{LiangINFOCOM2014, LiangTNET2014, ShahTCOM2016, ZhongINFOCOM2018, Behrouzi-FarICBD2019}. 
The shifted exponential distribution can be approximated to an exponential distribution when the constant shift is much smaller than the mean of the distribution. 
Since exponential distributions offer analytical tractability, 
we consider the case when the constant shift is negligible in empirical write distributions. 
Hence, we assume that the write times are \iid across write requests and servers, 
distributed exponentially with rate $\mu_w$. 

\subsection{Performance Metric}

For an $(n+1)$ server system, 
with arrival and service rate pairs $(\lambda_r, \mu_r)$ for read requests and $(\lambda_w, \mu_w)$ for write requests, 
we will measure the system performance by the limiting average of number of requests in the system. 
We denote aggregate read and write load on the system respectively by 
\begin{xalignat}{2}
\label{eqn:DefnLoad}
&\rho_r \triangleq {\lambda_r}/{\mu_r},&
&\rho_w \triangleq {\lambda_w}/{\mu_w}.
\end{xalignat} 
Since write requests join all $(n+1)$ servers, 
and read requests join one of the $(n+1)$ servers, 
the system is stable if $\rho_w + \rho_r/(n+1) < 1$. 
Let $M_n(t)$ denote the number of requests in the $(n+1)$-server system at time $t$, 
then the limiting average number of requests is 
\EQN{
\label{eqn:AvgNumRequests}
\bar{M}_n = \lim_{t \to \infty}\frac{1}{t}\int_{s=0}^tM_n(s)ds.
}

From Little's law~\cite{LittleIJOR1961},
we know that the limiting mean number of requests is directly proportional to the limiting mean sojourn time of requests in the system. 
Performance metric of choice here is the number of requests in the system, 
which is sum of the number of requests of individual classes. 
Implicitly, we have assumed that the read and write requests are of equal importance in our work. 
However, since we separately compute the number of read and write requests in the system, 
our framework can be used for any performance metric that is a function of the two numbers.  

We will show that the number of secondary servers $n$ is an important system parameter that controls the system performance. 
Specifically, we will show that under certain traffic and service parameters, 
there exists an optimal choice of number of secondary servers $n$, 
that minimizes the limiting average of number of requests in the system. 
Formally, we solve the following 
problem. 
\begin{prob}
\label{prob:OptRedundancy} 
For the distributed read-write system described above, 
find the optimal number of servers $n^\ast$ such that 
\EQ{
n^\ast = \arg\min_n\bar{M}_n.
}
In particular, we will find the optimal number of servers for read-first and write-first priorities. 
We will assume that the system is stable for all $n \in \Z_+$, i.e. $\rho_r +\rho_w < 1$. 
Further, we assume finite write load on the system, i.e. $\rho_w > 0$. 
\end{prob}

\section{Background}
\label{section:Background}
For our system model, read requests are easier to understand. 
Since read arrivals get routed to one of $(n+1)$ servers uniformly at random, 
the read requests queues would have remained independent if there were no write requests in the system. 
The write requests arrive at the primary, and 
are sent to all $n$ secondary servers at the instant of write completion at the primary. 
A request is considered completed, if it is completed at all $n$ servers. 
That is, a write request is forked to all $n$ secondary servers, and joins after service from all $n$ of them. 
This is precisely the setting of $(n,n)$ fork-join queues~\cite{BaccelliPCIS1985}.

In this section, we will just focus on the evolution of an $(n,n)$ fork-join system for a single class of requests. 
We study the impact of preceding primary server in Section~\ref{section:WritePriority} and~\ref{section:ReadPriority}. 
Specifically, we consider an $n$ server system,  
where request arrival is a Poisson process with rate $\lambda$, 
and the service time of each request at all servers is assumed to be \iid exponential with rate $\mu$. 
An arriving request is forked to all $n$ servers, 
and it leaves a server after service. 
Customers are served in an FCFS manner at each server. 

\subsection{Bounds using individual server queues}
\label{subsection:Bounds}
Let $X_j(t)$ be the number of  requests in the queue~$j$, 
then we make two observations. 
The evolution of each queue follows an $M/M/1$ queue
and the total number of requests in the system at time $t$ is given by $\max\set{X_j(t): j \in [n]}$. 
The limiting distribution of number of requests in each of the $n$ queues can be computed easily. 
If the queues were independent, then this also gives us the distribution of total number of requests in the system. 
However, since all queues get new requests at the identical arrival instant, 
they are not independent. 
Nevertheless, we do have bounds on the limiting mean number of requests. 
A simple lower bound is derived from the application of Jensen's inequality to convex function $\max$, 
and an upper bound is derived from the fact that $\max$ is upper bounded by the sum. 
That is, we have  
\EQ{
\max_{j \in [n]}\lim_{t \in \R_+}\E X_j(t) \le \lim_{t \in \R_+}\E\max_{j \in [n]}X_j(t) \le \sum_{j \in [n]}\lim_{t \in \R_+}\E X_j(t). 
}

Computational bounds on $(n,n)$ fork-join queuing system has been studied in~\cite{BaccelliADAP1989}, where the authors provide an upper and lower bound for $n$ server homogeneous systems with Poisson arrivals of rate $\lambda$ and \iid exponential service times with mean $1/\mu$.  
We cite these general bounds here, and will adapt them to write priority system in Section~\ref{subsection:WritePriorityLUB}. 
It was shown in~\cite{BaccelliADAP1989} that the lower bound on mean response time is achieved by $n$ parallel $D/M/1$ queues with deterministic periodic arrivals with mean inter-arrival time $1/\lambda$ and memoryless service time distribution identical to the original system.  
Applying the Little's law to the lower bound~\cite[Eq. 6.20]{BaccelliADAP1989} for exponential service times, 
\EQN{
\label{eqn:LowerBound}
\lim_{t\in\R_+}\E[\max_{j \in [n]} X_j(t)] \ge \frac{\lambda}{\mu(1-\eta)}H_n,
}
where  the Harmonic sum of first $n$ positive integers is $H_n \triangleq \sum_{i=1}^{n}\frac{1}{i}$,  
and $\eta$ is the smallest positive solution to the equation
\EQN{
\label{eqn:LBProb}
x = \exp(-\frac{\mu(1-x)}{\lambda}),\qquad x\ge 0.
}
The upper bound on mean response time for $(n,n)$ fork-join queue is achieved by $n$ independent $M/M/1$ queues with Poisson arrivals of rate $\lambda$ and \iid exponential service times with rate $\mu$. 
Applying Little's law to the upper bound~\cite[Eq. 6.29]{BaccelliADAP1989}, 
substituting the Laplace-Stieltjes transform of the distribution function of inter-arrival times as $A^\ast(s) = \frac{\lambda}{\lambda+s}$, 
and under unit loading for all arrivals, 
we get 
\EQN{
\label{eqn:UpperBound}
\lim_{t\in\R_+}\E[\max_{j \in [n]} X_j(t)]\le \frac{\lambda}{\mu(1-\delta)}H_n,
}
where the Harmonic sum $H_n = \sum_{i=1}^{n}\frac{1}{i}$ and $\delta = \frac{\lambda}{\mu}$ is the smallest positive solution to the equation
\EQN{
\label{eqn:UBProb}
x = A^\ast(\mu(1-x)) 
= \frac{\lambda}{\lambda+\mu(1-x)},
\qquad x\ge 0.
}
For exponential service times, the upper and lower bounds are of the same order in number of servers $n$, 
and only differ in the constant. 
We will see that the time to service write request in a read priority system with preemption is no longer memoryless, 
as one may have to service multiple read requests before one write request can be serviced.
In this case, computing the closed form for the write service time distribution is challenging. 
We further observe that the departure of write requests from primary server is no longer memoryless. 
Further, the write requests enter $(n,n)$ fork-join queue at $n$ secondary servers with general service times at each of them. 
Even if the arrival process to the secondary servers is approximated by a Poisson process, 
the lower bounding system for fork-join queues is $n$ independent $D/G/1$ queues, 
and the upper bounding system for fork-join queue is $n$ independent $M/G/1$ queues. 
The closed form computation of the mean of the maximum of $n$ such queues remains difficult, 
due to the difficulty of computing the closed form service time distribution of a write request. 

Therefore, we consider the alternative way of viewing fork-join system as tandem queues with pooled service as proposed in~\cite{BaditaTIT2019}. 
For fork-join queues with memoryless service, 
the authors of~\cite{BaditaTIT2019} proposed an independent approximation for the tandem queues that is shown to remain tight across all ranges of system parameters. 
We adapt this approach to study both write and read priority systems. 
We observe that under write priority system, the approximation matches the upper bound.  
We propose a new approximation for read priority system, 
under which the computed approximate mean number of write requests in the system is shown to be close to the empirical mean in the original system. 

\subsection{Tandem queue approach and approximation} 
\label{subsection:TandemQueue}
An alternative way of state representation of fork-join queues at any time $t$, 
is the sequence of set of servers that have served each request in the system~\cite{BaditaTIT2019}. 
Each incoming request is served in FCFS fashion at each server and is forked instantaneously to all $n$ servers. 
Therefore, the set of servers serving the newer requests are the ones that have already served the older requests. 
From the homogeneity in the system, it follows that the identities of servers do not matter, 
and the system state is sufficiently represented by the number of servers that have served each request in the system. 
From the FCFS service discipline, it follows that older requests are served by the number of servers no less than the newer requests. 
Therefore, one can partition all requests in the system by the number of servers that have served it.  
Accordingly, let $Y_i(t)$ denote the number of requests in the system that have been served by $i$ servers at time $t$, 
and denote $Y(t) \triangleq (Y_0(t), \dots, Y_{n-1}(t))$. 
It turns out that $(Y(t), t \ge 0)$ is a sufficient state representation for a fork-join queue. 
We call $Y_i(t)$ to be the number of requests in level~$i$. 
We observe that each incoming arrival is served by~$0$ servers, and hence this increases $Y_0(t) \to Y_0(t)+1$ at the arrival instant $t$. 
Further, when a request is served by all $n$ servers, it departs the system and decreases $Y_{n-1}(t) \to Y_{n-1}(t) - 1$ at the departure instant $t$.
When a request is served by $i < n$ servers, 
then it becomes a request with $i$ server completions from $(i-1)$ server completions. 
Correspondingly, we have $Y_{i-1}(t) \to Y_{i-1}(t)-1$ and $Y_i(t) \to Y_i(t)+1$ at this service completion instant $t$. 
That is, $Y(t)$ has a tandem queue interpretation with queues~$0$ to $(n-1)$ from left to right, 
with external arrivals to queue~$0$ and external departures from queue~$(n-1)$. 
A departure from queue~$(i-1)$ leads to an arrival to queue~$i$. 

For each partition of requests that have been served by $i$ servers, 
can only be served by remaining $(n-i)$ servers.  
At these $(n-1)$ server queues, there maybe requests that have been served by~$(i+1)$ or more servers.  
The requests with~$(i+1)$ or more completions are older than the requests with~$i$ completions, 
and hence are served first due to FCFS service discipline.  
We let $N_i(t)$ denote the number of servers at time $t$, 
whose head request has been served by exactly $i$ servers. 
That is, the requests in level~$i$ of the tandem queue, are served by $N_i(t)$ servers in parallel at time $t$. 
We have illustrated the tandem queue interpretation of the system in Figure~\ref{fig:TandemPooled},
where a request served by $i$ servers moves to level $i+1$ once it receives service from one of the $N_i(t)$ servers. 
\begin{figure}[h]
\centering
\begin{tikzpicture}
[font=\small, line width=1pt, node distance=0mm, draw=black, >=stealth',
server/.style={circle, minimum size=6mm, draw=black, fill=gray!40, inner sep=0pt},
cache/.style={rectangle, rounded corners=1mm, minimum height=6mm, minimum width=15mm, draw=none, inner sep=0pt}
]

\node[server, label=below:{$N_1(t)=1$}](serv1) at (0,0) {$\mu$};
\node[cache, left = of serv1] (cache1) {$Y_1(t)$};
\draw[rounded corners=1mm] (cache1.south west) -- (cache1.south east) -- (cache1.north east) -- (cache1.north west);

\node[server, left = 10mm of cache1, label=below:{$N_0(t)=1\SetIn{Y_1(t)>0}+2\SetIn{Y_1(t)=0}$}] (serv0) {$\mu$};
\node[cache, left = of serv0] (cache0) {$Y_0(t)$};
\draw[rounded corners=1mm] (cache0.south west) -- (cache0.south east) -- (cache0.north east) -- (cache0.north west);
  \draw (serv0) edge[->] (cache1);
  \draw ([xshift=-10mm]cache0.west) edge[->] node[above] {$\lambda$} (cache0);
  \draw (serv1) edge[->] ([xshift=10mm]serv1.east);
  \draw[dashed, ->, rounded corners=1mm] (serv1.north) -- ([yshift=5mm]serv1.north) -- ([yshift=5mm]serv0.north) -- (serv0.north);

\end{tikzpicture}
\caption{Write requests in the system represented as a tandem queue with pooled servers.}
\label{fig:TandemPooled}
\end{figure}
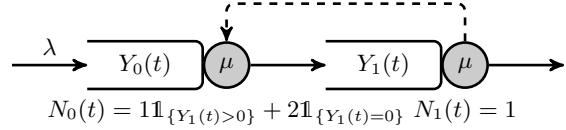
\begin{prop}
\label{prop:NumUsefulServers}
For the $(n,n)$ fork-join queuing system under consideration,  
the vector $Y(t)$ represents the occupancy of an $n$-tandem queue at time $t$.  
If $Y(t) = y$, then the number of servers serving the head request at $i$th tandem queue is 
denoted by $N_i(t) = N_i(y)$ such that 
\EQ{
N_i(y) = \begin{cases}
1, & i = n-1,\\
1 + N_{i+1}(y)\SetIn{y_{i+1} = 0}, &i < n-1.
\end{cases}
}
\end{prop}

\begin{proof}
We just provide a proof sketch here, which is adapted from~\cite{BaditaTIT2019}.
Each request joins all $n$ queues, gets served in an FCFS manner at each of them, 
and leaves the system when it has been served by all $n$ servers. 
When a request gets serviced at a server, we say that the corresponding server has been observed by the request. 
Due to FCFS service policy, the subset of observed servers for each request in the system forms a chain~\cite{BaditaTIT2019},
i.e. the set of observed servers for an older request in the system contains the set of observed servers of newer request. 
Thus, we can aggregate all the requests with identical set of observed servers. 
Then the number of requests with identical set of $i$ observed servers is denoted by $Y_i(t)$ at time $t$. 
We note that because of FCFS service policy, only the oldest of the $Y_i$ requests is in service. 
Since requests leave after observing $n$ servers, we have $i\in \set{0, \dots, n-1}$ and the number of observed servers $i$ is referred to as level~$i$. 
One can see that $Y$ is the number of requests in a tandem queue with $n$ levels. 
Incoming requests have not observed any server and hence join level~$0$. 
When a request with $i$ observed servers, gets served by another useful server, 
it leaves level~$i$, and joins level~$i+1$. 
A request departs the system from level~$n-1$.  

For an $(n,n)$ fork-join system, the set of useful servers for a request is the set of all servers without the observed servers.  
The number of useful servers for the requests in level~$i$ is $n-i$. 
However, some of these servers are useful to the requests ahead of them as well, due to chain property of observed servers. 
The number of available servers for requests with $i$ observed servers is denoted by $N_i(t)$ at time $t$. 
For $i = n-1$, the number of useful and available servers remain same and equal to $1$, 
since requests with $n$ observed servers leave. 
A request with $i$ observed servers has one available server not useful to requests with $(i+1)$ observed servers, 
and it can use the available servers from the $(i+1)$th level, if there are no requests with $(i+1)$ observed servers. 
Since, this relation only depends on time $t$ through the occupancy vector $Y$, 
we get the result. 
\end{proof}

\begin{prop}
\label{prop:ForkJoinCTMC}
For the $(n,n)$ fork-join queuing system described above, 
the occupancy vector $(Y(t): t \in \R_+)$ forms a continuous-time Markov chain, with possible transitions 
\EQ{
a_i(y) \triangleq 
\begin{cases}
y +e_0,& i = 0,\\
y+e_i-e_{i-1}, & i \in [n-1], y_{i-1} > 0,\\
y - e_{n-1}, & i= n, y_{n-1} > 0,
\end{cases}
}
and the corresponding transition rates are 
\EQ{
Q(y,a_i(y)) = 
\lambda\SetIn{i=0} 
+ N_{i-1}(y)\mu\SetIn{y_{i-1} > 0}\SetIn{i \in [n]}.
}
\end{prop}
\begin{proof}
In the proof of Proposition~\ref{prop:NumUsefulServers}, 
we observed that the occupancy vector $Y(t)$ has a tandem queue interpretation. 
We will show that the random time to next transition depends only on the current state, 
and has an exponential distribution. 
Further, conditioned on the current state, 
the jump probability of next transition is independent of the random transition time, and previous jump transitions. 
This shows that process $Y$ is a continuous-time Markov chain~\cite{RossWiley2008}.

Recall that there are three types of possible transitions from the current state $y$, 
namely an external Poisson arrival to level~$0$, 
a service for request with~$(i-1)$ observed servers that leads to departure of this request from level~$(i-1)$ and arrival to level~$i$ for $i < n$, 
and a service for request with $(n-1)$ observed servers that leads to an external departure from the system. 
Recall that the number of servers available to level~$(i-1)$ is $N_{i-1}(y)$, 
each server has an \iid exponential service time with rate $\mu$, 
and the inter-arrival times for external arrivals are \iid exponential with rate $\lambda$. 
Therefore, the residual times are independent and exponentially distributed.

The time for next transition is minimum of all these random times. 
Conditioned on the current state, next inter-transition time is an exponential random variable independent of the past, 
with rate given by the sum $\lambda + \mu\sum_{i=1}^nN_{i-1}(y)\SetIn{y_{i-1} > 0}$. 
Further, the probability that one of these transitions take place is given by the ratio of the transition rate and the sum-rate. 
This probability is independent of the inter-transition time and past transitions, given the current state $y$. 
\end{proof}
This Markov process is interpreted as a tandem queue, as shown in Figure~\ref{fig:TandemPooled}, 
where each level~$i$ has its own dedicated service rate $\mu$,
which gets pooled in the levels below when level $i$ is empty.
To compute the mean number of write requests, we need to find the invariant distribution of the Markov process $Y$. 
However, this problem 
remains intractable since it is equivalent to finding an eigenvector of an $n$-dimensional operator with eigenvalue unity. 
Further, the Markov process $Y$ is not reversible, and hence there are no known techniques to find the invariant distribution of this process. 
However, a tight reversible approximation of this process was proposed in~\cite{BaditaTIT2019}, 
which we quote here for reference. 
\begin{appr}
\label{appr:Rate} 
The pooled tandem queue $Y$ with dedicated rates $(\mu, \dots, \mu)$ is approximated by a continuous-time Markov process $\bar{Y} \triangleq (\bar{Y}(t) \in \Z_+^n, t \in \R_+)$, which is an unpooled tandem queue with Poisson arrival rate $\lambda$ and exponential service rates $(\bar{\gamma}_0, \dots, \bar{\gamma}_{n-1})$, 
where the service rate for level~$i$ in unpooled tandem queue is 
\EQN{
\label{eqn:ApproxTandemQueueRate}
\bar{\gamma}_i \triangleq (n-i)\mu - (n-i-1)\lambda.
}
\end{appr}

\section{Write priority system}
\label{section:WritePriority}
For systems with strict consistency requirements, write requests are prioritized over the read requests. 
Example of such systems are banking, financial services, e-commerce, military applications, etc. 
Let $R_j(t)$ denote the number of read requests at server $j \in \set{0, \dots, n}$ at time $t\in \R_+$, 
and $W(t)$ denote the number of unique write requests in the write priority system at time $t \in \R_+$. 
We will denote the stationary limit of the corresponding numbers by $R_j$ and $W$ respectively. 

Recall that write requests are first served by the primary server~$0$, and then forked to $n$ secondary servers. 
The write request is considered completed only when the request is written to all $(n+1)$ servers. 
Since write requests have preemptive priority, they are oblivious of the presence of read requests. 
As such, the number of write requests in such a system can be modeled by the number of requests in a system with single queue $W_0(t)$ followed by an $(n,n)$ fork-join queue.  
As seen in Section~\ref{section:Background}, 
an $(n,n)$ fork-join queue is equivalent to a sequence of $n$ pooled tandem queues denoted by $Y$. 
The number of write requests in the system is the sum of write requests in the primary queue and the 
$n$ tandem queues. 
We use Approximation~\ref{appr:Rate} to approximate the $n$ pooled tandem queues $Y$ by independent $n$ unpooled $M/M/1$ tandem queues $\bar{Y}$, 
and compute the mean number of write requests for the approximate system. 
For the approximate system, each queue~$i$ has a Poisson arrival of rate $\lambda$, service rate $\bar{\gamma}_i$, 
and we denote the stationary limit of the aggregate number of write requests in the system by $\bar{W} = W_0 + \sum_{i=0}^{n-1}\bar{Y}_i$. 
\begin{thm}
\label{thm:NumWriteRequests}
In terms of the harmonic sum $H_n \triangleq \sum_{i=1}^n\frac{1}{i}$, 
the mean number of write requests in the approximate system is
\EQ{
\E\bar{W}
= \frac{\rho_w}{1- \rho_w}(1 + H_n).
}
\end{thm}
\begin{proof}
Result follows from the linearity of expectation, 
the fact that the mean number of requests in an $M/M/1$ queue with arrival rate $\lambda$ and service rate $\bar{\gamma}_i$ is $\lambda/(\bar{\gamma}_i-\lambda)$, 
and the definition of $\bar{\gamma}_i$ in Eq.~\eqref{eqn:ApproxTandemQueueRate} for $\lambda = \lambda_w, \mu = \mu_w$. 
\end{proof}
\begin{rem}
\label{rem:writeRequests}
To show the explicit dependence of mean number of write requests in the approximate system, 
on the number of redundant secondary servers, 
we denote $f(n) \triangleq \E\bar{W}$. 
Recall that the harmonic sum $H_n = \sum_{i=1}^n\frac{1}{i}$ 
can be tightly approximated by its lower bound $\ln(n+1)$. 
Therefore, we can write the mean number of write requests in the approximate system as
\EQ{
f(n) \approx \frac{\rho_w}{1 - \rho_w}\ln e(n+1).
}
\end{rem}
We now focus on the mean sojourn time of lower priority read requests in the system. 
Recall that arrival processes of read and write requests are independent and Poisson with rates $\lambda_r$ and $\lambda_w$ respectively. 
Due to independent and uniform splitting, the arrival of read requests at a server~$j$ is Poisson with rate $\frac{\lambda_r}{(n+1)}$, 
and is independent of read arrival process at other servers.
Each write request is served by each of the $(n+1)$ servers, where an arriving write request joins server~$0$ and then upon service joins all remaining $n$ servers at once. 
Due to exponential service at server~$0$, the write arrival process at each server~$j$ is Poisson with rate $\lambda_w$,
and the arrivals at servers $(1, \dots, n)$ are dependent. 
Nevertheless, due to independence of read and write arrival processes, and due to priority of write service, the arrival of read and write requests at each server~$j$ remains independent. 

At each server~$j$, 
write requests have a priority over the read requests, and are served preemptively. 
We observe that the arrival process, service process, and scheduling of the read and write requests, 
is identical at each server. 
Hence, the evolution of queue occupancy at each server is identical at all $(n+1)$ servers.\footnote{The queue occupancy are not independent though, due to coupling of write request arrival instant at all secondary servers.} 
In particular, this implies that 
\EQN{
\label{eqn:MarginalMeanReadRequests}
\E R_0 = \dots = \E R_n. 
}
Since, the number of read requests in the system is the sum of read requests at the $(n+1)$ servers in the system, 
it follows from the linearity of expectation and Eq.~\eqref{eqn:MarginalMeanReadRequests}, 
that $\sum_{j=0}^n\E R_j = (n+1)\E R_0$. 

Therefore, we focus on the single server queue with two classes of independent arrivals, read and write, 
where write has preemptive priority over read. 
In particular, we are interested in finding the mean number of read requests at server~$0$ at stationarity. 
The service time of a read and write request is identically distributed to exponential random variables $S_r$ and $S_w$ with means ${1}/{\mu_r}$ and ${1}/{\mu_w}$ respectively. 
\begin{thm}
\label{thm:NumReadRequests}
The mean number of read requests at server~$0$ of the write priority system is given by 
\EQ{
\E R_0 
= \frac{\rho_r}{(n+1)(1-\rho_w)-\rho_r}\left(1 + \frac{\mu_r}{\mu_w}\frac{\rho_w}{(1-\rho_w)}\right).
}
\end{thm}
\begin{proof}
Due to PASTA~\cite{WolffIJOR1982},
an incoming read request sees $W_0$ write and $R_0$ read requests in the system, at stationarity. 
This arrival needs $S_r$ amount of service, and stays $T_r$ amount of time in the system. 
During the stay duration $T_r$, the mean number of write requests that arrive in the system is given by $\lambda_w \E T_r$. 
These additional write requests get serviced ahead of the read request, due to preemptive priority. 
Hence, we can write  
\EQ{
\E T_r = \E S_r + \E R_0 \E S_r + \E W_0 \E S_w + \lambda_w\E T_r\E S_w.
}
From Little's law~\cite{LittleIJOR1961},
we have $\E T_r = \frac{(n+1)}{\lambda_r}\E R_0$. 
Further, for the write requests each server is just an $M/M/1$ queue, and hence $\E W_0 = \frac{\rho_w}{1-\rho_w}$. 
Applying these results to the above equation, and re-arranging terms, we get the result. 
\end{proof}
\begin{rem}
\label{rem:readRequests} 
We first note that the mean number of read requests in the system remains exact for preemptive priority, 
as opposed to the mean number of write requests computed for an approximate system. 
To show the explicit dependence of mean number of read requests on the number of redundant secondary servers, 
we denote $g(n) \triangleq (n+1)\E R_0$. 
We observe that the mean number of read requests in the write priority system is decreasing with number of redundant servers $n$. 
Defining $\alpha \triangleq \frac{\mu_r}{\mu_w}$, and approximating $\frac{1}{1-x}\approx 1+x$, 
we approximate
\EQ{
g(n) \approx \frac{\rho_r}{(1-\rho_w)}\Big(1+\alpha\frac{\rho_w}{1-\rho_w}\Big)\Big(1+\frac{\rho_r}{(1-\rho_w)(n+1)}\Big).
} 
\end{rem}
\subsection{Optimal number of servers}
Recall that latency is proportional to the number of requests in the system, 
and hence we find the optimal number of servers $n^\ast$ that minimizes the limiting mean number of requests in the system, 
as defined in Problem~\ref{prob:OptRedundancy}. 
\begin{thm}
\label{thm:existence}
For a stable read-write system under write priority, 
there exists an optimal number of redundant secondary servers $n^\ast$ defined in Problem~\ref{prob:OptRedundancy}. 
The optimal redundancy is zero if
\EQ{
\Big(\frac{\rho_r}{1-\rho_w}\Big)^2\Big(\frac{1-\rho_w}{\rho_w}+\alpha\Big)<1.
}
\end{thm}
\begin{proof} 
Recall that, for an $(n+1)$ server read-write system, 
$f(n)$ denotes the mean number of write requests in the approximate system, 
and $g(n)$ denotes the mean number of read requests.  
We have established that $f$ is a logarithmically increasing function of $n$ in Remark~\ref{rem:writeRequests}. 
Further, we saw in Remark~\ref{rem:readRequests}, 
that $g$ is decreasing as $\frac{1}{n}$. 
This implies that, there always exists a unique minimum of the sum $f+g$. 
Treating $f(n)$ and $g(n)$ as evaluation of continuous functions at integer values, 
we can differentiate $f+g$ and find the unique real number $x^\ast$ such that $f'(x^\ast)+g'(x^\ast)=0$.
Denoting $\nu \triangleq \frac{\rho_w}{1-\rho_w}$, we can write $x^\ast 
=(\rho_r(1+\nu))^2\Big(\frac{1}{\nu} + \alpha \Big)-1.$
We verify that the mean number of requests increases with $n$ for $n>x^\ast$, 
and hence the optimal number of secondary servers is zero when $x^\ast < 0$.  
\end{proof}
\begin{rem}
When $x^\ast > 0$, 
we can find the optimal number of redundant servers $n^\ast$ by comparing mean number of requests $f+g$ at integer values $\ceil{x^\ast}$ and $\floor{x^\ast}$, i.e.
\EQ{
n^\ast \triangleq \arg\min\set{(f+g)(\floor{x^\ast}), (f+g)(\ceil{x^\ast})}. 
}
From the equation for $x^\ast$, we conclude that the optimal redundancy $n^\ast$ increases with read load $\rho_r$. 
Further, when the read service rate is larger than the write service rate\footnote{Typically, reads are faster than writes, and $\alpha = \frac{\mu_r}{\mu_w} > 1$.}, the optimal redundancy decreases with write load until write load $\rho_w$ is smaller than a threshold $\rho_T$,\footnote{The write load threshold is $\rho_T \triangleq \min\lbrace\frac{-3+\sqrt{9+8(\alpha-1)}}{4(\alpha-1)},1-\rho_r\rbrace$.} and increases thereafter. 
\end{rem}
\textbf{Design principle}:
Since the mean number of write requests grows logarithmically in $n$,
while the mean number of read requests decreases as ${1}/{n}$,
the right design choice for number of servers in a write priority system depends on the load and the service rates. 

\subsection{Lower and upper bound}
\label{subsection:WritePriorityLUB}
Under write priority, the departure process of a write request from the primary server is a Poisson process of rate $\lambda$ at stationarity.  
Subsequently, a write request joins an $(n,n)$ fork-join system of secondary servers, 
with Poisson arrivals of rate $\lambda$ and \iid memoryless rate $\mu$ service times. 
From Section~\ref{subsection:Bounds}, 
it follows that the mean number of requests at $n$ secondary servers is lower bounded by $\frac{\rho_w}{(1-\eta)}H_n$ where $\eta$ is the smallest positive solution for Eq.~\eqref{eqn:LBProb} with $\mu=\mu_w$ and $\lambda=\lambda_w$. 
Adding the mean number of read requests, 
the mean number of write requests at the primary server, 
and the lower bound on the mean number of write requests at $n$ secondary servers, 
we can lower bound the total number of requests in the system as
\EQ{
\E\bar{W} + (n+1)\E R_0\ge \rho_w\Big(\frac{1}{1-\rho_w} + \frac{H_n}{1-\eta}\Big) + g(n).
}
Substituting write Poisson arrival rate $\lambda=\lambda_w$ and exponential service rate $\mu=\mu_w$ in Eq.~\eqref{eqn:UBProb}, 
we obtain an upper bound $\frac{\rho_w}{(1-\rho_w)}H_n$ on the mean number of requests at $n$ secondary servers.  
Adding the mean number of read requests, 
the mean number of write requests at the primary server, 
and the upper bound on the mean number of write requests at $n$ secondary servers, 
we can upper bound the total number of requests in the system as
\EQ{
\E\bar{W} + (n+1)\E R_0\le \frac{\rho_w}{1-\rho_w}(1 + H_n) + g(n) = f(n)+g(n).
}
We note that  for the write priority system, 
the upper bound coincides with the approximation obtained earlier.

\section{Read priority system}
\label{section:ReadPriority} 

For systems where consumers are not sensitive to the timeliness of data, 
but sensitive to the read latency, read requests are prioritized. 
Examples of such systems are video streaming, cataloging, data mining, content management systems, etc. 
In this section, we analyze the distributed read-write systems with priority for read requests. 
Recall that incoming read requests are routed to one of the $(n+1)$ servers, uniformly at random.  
At any server in a read priority system, any write request in service is preempted by an incoming read request.  
Let $R_j(t)$ be the number of read requests at server~$j$ at time $t$. 
The evolution of $R_j(t)$ is equivalent to an $M/M/1$ queue with Poisson arrivals of rate ${\lambda_r}/{(n+1)}$ and exponential service of rate $\mu_r$. 
The read load on any server~$j$ is ${\rho_r}/{n+1}$. 
\begin{rem}
To show the explicit dependence of mean number of read requests on the number of redundant secondary servers, 
we denote this mean by $q(n)$ in the read priority system. 
Due to linearity of expectation and the evolution of $R_j(t)$ at all servers as identical $M/M/1$ queues,
we can write the mean number of read requests in the read priority system as 
\EQ{
q(n) \triangleq \sum_{j=0}^n\E R_j = (n+1)\E R_0 = \frac{(n+1)\rho_r}{n+1-\rho_r}.
}
We observe that the mean number of read requests in the read priority system is decreasing in number of redundant servers $n$, 
and saturates to the read load $\rho_r$.
\end{rem}

We next focus on the number of unique write requests in the system, 
denoted by $W(t)$ at time $t$. 
Recall that write requests initially join the primary server, 
and upon service completion join all the $n$ secondary servers instantaneously. 
A write request leaves the system, when it has been served by all the servers. 
Since a write request is preempted when an incoming read request arrives at that server, 
we cannot use Approximation~\ref{appr:Rate} for the write process as it is coupled with number of read requests in the system.  

Let $W_0(t)$ denote the number of write requests in the primary server at time $t$, 
then $W(t)-W_0(t)$ denotes the number of write requests forked at $n$ secondary servers. 
Service for forked write requests at server~$j$ depends on the number of read requests $R_j(t)$ at this server as well. 
For each write request $k \in [W(t)-W_0(t)]$ ordered by their arrival time, 
we define $S_k(t)\subset[n]$ to be the set of secondary servers that has already served this write request at time $t$. 
Note that the write request exits the system after being served by all $n$ secondary servers. 
Since the system follows FCFS policy, the later arrivals only get served by the servers that have already served the previous arrivals. 
In particular, we observe that the $(k+1)$th request cannot be served by a server before it has served the $k$th request.
Therefore, 
\EQN{
\label{eqn:structure}
\emptyset\subseteq S_{k+1}(t)\subseteq S_k(t)\subset [n],\quad k, k+1 \in [W(t)-W_0(t)].
}
We write the sequence of set of observed servers for all forked write requests in the system as $S(t) \triangleq (S_k(t): k \in [W(t)-W_0(t)])$, 
and we denote the number of read requests at $(n+1)$ servers by $R(t) \triangleq (R_j(t): 0 \le j \le n)$.  
Then, we can denote the state of the system at time $t \in \R_+$ by $Z(t) \triangleq (W_0(t), S(t), R(t)) \in \sZ \triangleq \Z_+\times [n]^\ast \times \Z_+^{\set{0, \dots, n}}$. 

\begin{thm}
\label{thm:WriteProcMC}
For a distributed read-write system with priority for read requests, 
the random process $(Z(t), t \in \R_+)$ forms a continuous-time Markov chain. 
For a state $z = (w_0, s, r) \in \sZ$, defining $s_k^j \triangleq (s_1, \dots, s_k\cup\set{j}, \dots)$ and $s^\prime \triangleq (s_2, \dots, s_{\abs{s}})$, 
the associated generator matrix $Q$ is given by 
\eq{
Q(z,z^\prime) &= \frac{\lambda_r}{n+1}\sum_{j=0}^n\SetIn{z^\prime = (w_0,s,r+e_j)} +\mu_r\sum_{j=0}^n\SetIn{z^\prime = (w_0,s,r-e_j)}\SetIn{r_j \ge 1} +\lambda_w\SetIn{z^\prime = (w_0+1,s,r)}\\ 
&+\mu_w\SetIn{z^\prime = (w_0-1,(s,\emptyset),r)}\SetIn{r_0=0}\SetIn{w_0 \ge 1}
+\mu_w \Bigg(\sum_{j: r_j=0}\sum_{k > 1}\SetIn{z^\prime = (w, s_k^j, r)}\SetIn{j \in s_{k-1}\setminus s_k}\\
&+\SetIn{j \notin s_1}\Big(\SetIn{z^\prime = (w, s_1^j, r)}\SetIn{\abs{s_1} < n-1}
+\SetIn{z^\prime = (w, s^\prime, r)}\SetIn{\abs{s_1} = n-1}\Big)\Bigg).
}

\end{thm}
\begin{proof} 
We observe that system state can change if there is (a) an external arrival of read request, (b) a read request gets serviced, (c) an external arrival of write request, and (d) a write request gets serviced. 
Since all distributions are continuous and arrival and service times are independent, 
only one of these events take place in an infinitesimal time. 
We will first focus on state transitions due to read requests, 
since their evolution is easier to understand when they are prioritized. 

\textbf{(a) External arrival of a read request:} 
External read arrival at each server is an independent Poisson process of rate $\frac{\lambda_r}{n+1}$ .
Therefore, an external arrival of a read request to server $j \in [n]$ changes the state $(w_0,s,r) \to (w_0, s, r+e_j)$, 
and the inter-transition time for such transitions are independent for all $j \in [n]$ and distributed exponentially with rate $\frac{\lambda_r}{(n+1)}$. 

\textbf{(b) Service of a read request:}
Service time for read requests at each server is independent and memoryless with rate $\mu_r$, 
and hence the inter-transition time for transitions $(w_0, s,r) \to (w_0,s,r-e_j)$ are independent and memoryless with rate $\mu_r$ for all $j \in [n]$ such that $r_j \ge 1$. 

We next focus on the evolution of write requests in the system. 
Specifically, we look at the arrival and service of write requests. 

\textbf{(c) External arrival of a write request:}
External arrival of write requests is Poisson with rate $\lambda_w$, 
and the incoming write requests initially join the primary server. 
This leads to an increase in number of write requests $w_0$ at the primary. 
Thus, the inter-transition time for transitions $(w_0,s,r) \to (w_0+1,s,r)$ are independent and memoryless with rate $\lambda_w$. 

\textbf{(d) Service of a write request:} 
We first focus on service of an existing write request at primary server~$0$, 
which can only happen when there are no read requests at the primary server. 
This request departs from the primary server upon service completion, and is forked to all $n$ secondary servers. 
This request has not been served by any secondary servers at this instant. 
Since the service time for write requests are \iid and memoryless with rate $\mu_w$, 
the inter-transition time for transitions $(w_0,s,r) \to (w_0-1,(s, \emptyset),r)$ are independent and memoryless with rate $\mu_w$, and they occur when $r_0 = 0$ and $w_0 \ge 1$.

We next focus on servers $j \in [n]$ without any read requests, that can serve write requests.  
Recall that service time for write requests at each server is \iid and memoryless with rate $\mu_w$, 
and the sequence of set of secondary servers that have finished serving existing $\abs{s}$ write requests are $(s_1, \dots, s_{\abs{s}})$. 
From the monotonicity of these sets in Eq.~\eqref{eqn:structure} due to FCFS scheduling, 
server~$j$ can serve request $k \ge 2$, only if it has already served first $(k-1)$ requests, 
and hence $j \in s_{k-1}\setminus s_k$. 
This service leads to $k$th request getting served by server $j$, 
and it follows that the inter-transition time for transitions $(w_0,s,r) \to (w_0, s_k^j, r)$ are independent and memoryless with rate $\mu_w$, 
for all $j \in [n]$ such that $r_j = 0, j \in s_{k-1} \setminus s_k$, and $k \ge 2$. 
A server~$j$ can serve first existing request in the system, if it has not served it already and this service can lead to an external departure if $s_1\cup\set{j} = [n]$. 
It follows that the inter-transition time for transitions $z \to (w_0, s_1^j, r)$ and $z \to (w_0, (s_2, \dots, s_{\abs{s}}), r)$ are independent and memoryless with rates $\mu_w\SetIn{\abs{s_1}<n-1}$ and $\mu_w\SetIn{\abs{s_1}= n-1}$ respectively, 
for servers $j \in [n]$ such that $r_j = 0$ and $j \notin s_1$. 

Since each of these transitions are independent and memoryless, 
it follows that the process $(Z(t), t \in \R_+)$ is a continuous-time Markov chain~\cite{RossWiley2008},
and we have obtained the generator matrix for this Markov process. 
\end{proof}
\subsection{Approximate Markov Chain}
For the distributed read-write system with prioritized reads with service preemption, 
the number of read requests in the system $(R(t): t \in \R_+)$ forms a continuous-time Markov chain.
Recall that the evolution of $(R_j(t): t \in \R_+)$ is governed by an $M/M/1$ queue for all $j\in\{0,1,\dots,n\}$,
and the read load on each server is $\frac{\rho_r}{n+1}$.
Hence, the probability of a server $j$ not having any read request is given by $1-\frac{\rho_r}{n+1}$. 
Writing the number of write requests that have been served by $i$ secondary servers as $Y_i(t)$, 
we observe that 
\EQ{
Y_i(t) = \abs{\set{k \in [W(t)-W_0(t)]: \abs{S_k(t)} = i}}.
}
As seen in Eq.~\eqref{eqn:structure}, 
the sequence of observed servers $(S_k(t): k \in [W(t)-W_0(t)])$ for all write requests in the system are ordered by set inclusion.  
Hence, the write requests with $i$ service completions are served by the same set of servers. 
From Proposition~\ref{prop:NumUsefulServers}, 
it follows that $Y(t) \triangleq (Y_0(t), \dots, Y_{n-1}(t))$ is a pooled tandem queue, 
and evolves as a continuous-time Markov chain if there were no read requests in the system. 
This pooled tandem queue is approximated by an uncoupled tandem queue in Approximation~\ref{appr:Rate}, 
where the service rate of $i$th tandem queue is $\bar{\gamma}_i$.  
In read priority system, the evolution of $Y(t)$ also depends on the number of read requests in the system at each individual server, and the set of servers which are serving $i$th tandem queue. 

\begin{rem}
In read priority systems, 
a secondary server can only serve a write request if there are no read requests at this server. 
Since the probability of having zero read requests at a server is same for all servers, 
we can approximate the expected rate at which the $i$th stage of write tandem queue is being served as  $P\set{R_j=0}\bar{\gamma}_i = (1-\frac{\rho_r}{n+1})\bar{\gamma}_i$, where $\bar{\gamma}_i$ is defined in Eq.~\eqref{eqn:ApproxTandemQueueRate}.  
Similarly, the average service rate at the primary queue is $(1-\frac{\rho_r}{n+1})\mu_w$.
\end{rem}
That is, we will approximate the process $(Y(t): t \in \R_+)$ with process $(\bar{Y}(t): t \in \R_+)$, 
where the $\bar{Y}(t)$ is an uncoupled tandem queue with no read requests, 
and the service rates of each stage is multiplied by the probability of no read request at any server. 
\begin{appr}
\label{appr:WriteForkingReadPriority} 
For the distributed read-write system with prioritized preemptive reads, 
the number of write requests in the system can be modeled by a sequence of uncoupled tandem queues $(W_0(t), \bar{Y}_0(t), \dots, \bar{Y}_{n-1}(t): t \in \R_+)$ that are served at memoryless service rates $(\mu_0, \beta_0, \dots, \beta_{n-1})$.  
The first queue $W_0$ is served at rate $\mu_0 \triangleq \mu_w(1-\frac{\rho_r}{n+1})$ and the $i$th stage of tandem queue $Y_i$ is served at rate $\beta_i \triangleq \bar{\gamma}_i(1-\frac{\rho_r}{n+1})$, 
where $\bar{\gamma}_i$ is defined in Eq.~\eqref{eqn:ApproxTandemQueueRate} for $\lambda = \lambda_w, \mu = \mu_w$.  
\end{appr}

\begin{thm}
\label{thm:rPriorNumWriteRequests}
The mean number of write requests in the approximate system is
\EQ{
\E \bar{W} = \frac{\lambda_w}{\mu_0 - \lambda_w} + \sum_{i=1}^{n}\dfrac{\lambda_w}{\beta_{n-i} - \lambda_w}.
}
\end{thm}
\begin{proof}
Each of the $n$ unpooled tandem queues is an $M/M/1$ queue with arrival rate $\lambda$ and service rate $\beta_i$.  
Therefore, the mean number of request in $i$th queue is $\frac{\lambda}{\beta_i - \lambda}$.
Since the number of write requests in the system is the sum of requests in the primary queue and the $n$ tandem queues, 
the result follows. 
\end{proof}

\begin{rem} 
As in previous section, 
to show the explicit dependence of mean number of write requests on the number of redundant secondary servers, 
we denote $p(n) \triangleq \E\bar{W}$,
in the approximate read-write system with read priority. 
The expression for the mean number of write requests in the approximate read priority system can be written in terms of the write load parameter $\nu\triangleq{\rho_w}/{(1-\rho_w)}$, 
the read load parameter $\Delta_n\triangleq 1-\frac{\rho_r\nu}{(n+1-\rho_r)}$, 
and the digamma function\footnote{Digamma function $\psi(.)$ is the derivative of the logarithm of gamma function, and is continuous and differentiable in $\R_+$. It has the property $\psi(z+1) = \psi(z) + \frac{1}{z}$~\cite{AbramowitzDover1965}.
} $\psi: \bbC \to \R$ as
\EQ{
p(n) = \frac{(n+1)\nu}{n+1-\rho_r(1+\nu)} + \frac{(n+1)\nu(\psi(\Delta_n+n)-\psi(\Delta_n))}{(n+1-\rho_r)}.
} 
\end{rem}
\begin{rem}
\label{rem:Unbounded} 
We observe that as the number of redundant servers $n$ increases, 
the mean number of read requests reaches a finite limit $\rho_r$, 
and the mean number of write requests grows to infinity. 
\end{rem}
\begin{thm}
\label{thm:RPexistence}
For a stable read-write system under read priority, 
there exists an optimal number of redundant secondary servers $n^\ast$ defined in Problem~\ref{prob:OptRedundancy}. 
The optimal redundancy is non-zero if
\EQ{
2-\frac{\rho_r}{1-\rho_r-\rho_w}<\frac{\rho_r^2}{1-\rho_r}\Big(\frac{1}{\rho_w}-\frac{2}{2-\rho_r}\Big).
}
\end{thm}
\begin{proof}
The mean number of requests $\bar{M}$ in the read priority system is $p+q$, a function of number of redundant servers $n$.  
Treating $p(n)$ and $q(n)$ as evaluations of continuous functions at integer values, 
we observe that $\bar{M} = p+q$ is continuous and differentiable in $\R_+$. 
The function $\bar{M}$ grows unboundedly with $n$, 
as pointed out in Remark~\ref{rem:Unbounded}. 
Thus, for all positive $\epsilon > 0$, there exists a positive integer $m>1$ such that $\bar{M}_m = \bar{M}_1 + \epsilon$. 
By mean value theorem~\cite{ComenetzWScientific2002},
we know that there exists a positive real $c_2\in(1,m)$
such that value of the derivative at $c_2$ is $\frac{\epsilon}{m-1}$.
Similarly, there exists a value $c_1\in(0,1)$ such that value of the derivative at $c_1$ is $\bar{M}_1-\bar{M}_0$. 
When $\bar{M}_0 > \bar{M}_1$, there exists a positive real $x^\ast\in(c_1,c_2)$
such that $\bar{M}'(x^\ast) = 0$. 
This follows from the fact that the value of derivative of $\bar{M}$ at $c_1$ is negative, 
the value of derivative of $\bar{M}$ at $c_2$ is positive, 
and the derivative $\bar{M}'$ is continuous in $\R_+$. 
The result follows from rewriting this condition $\bar{M}_0 > \bar{M}_1$.
\end{proof}
\begin{rem}
Since the mean number of requests $p+q$, has terms with digamma function,
we numerically find the minimum $x^\ast$ in its domain $\R_+$.
Further, the optimal number of redundant servers $n^\ast$ can be found by comparing mean number of requests $p+q$
at integer values $\ceil{x^\ast}$ and $\floor{x^\ast}$, i.e.
\EQ{
n^\ast \triangleq \arg\min\set{(p+q)(\floor{x^\ast}), (p+q)(\ceil{x^\ast})}.
}
\end{rem}
\textbf{Design Principle :}
We observe that $p(n)$ is lower bounded by $\nu\ln e(n+1)$ for all $n\in\Z_+$,
and hence the asymptotic growth of the mean number of write requests is at least logarithmic in $n$.
Further, the number of read requests decreases in the order of $\frac{1}{n}$ with $n$.
Hence, this opposing behaviour of $p(n)$ and $q(n)$ should be taken into account while designing the system.

\section{Numerical Studies}

We numerically simulate a distributed read-write system with primary secondary architecture in this section. 
The system is simulated with non-preemptive scheduling and round-robin routing of read requests, 
as opposed to the simplifying assumption of preemptive scheduling and random routing, 
which we used for analysis. 
We compare the optimal number of servers obtained analytically using our proposed approximation, 
to the one observed empirically in the system simulation.
In our simulation studies, we select the read service rate $\mu_r=10$ an order higher than the write service rate $\mu_w=1$. 
This is motivated by the observation that reads are typically faster than writes in practical systems~\cite{KabakusCIS2017}. 
We have plotted the optimal number of servers for write and read priority systems, 
in Fig.~\ref{fig:WPnOpt} and Fig.~\ref{fig:RPnOpt} respectively. 
For both priorities, 
we fix arrival rate of the read or write requests, and vary the other within 95\% of the stability region.
The dashed curve shows the optimal number obtained analytically for the approximate system, 
while the solid curve denotes the optimal number obtained empirically under the system simulation. 

\subsection{Write priority system}
\label{subsection:WPNumericalStudies}
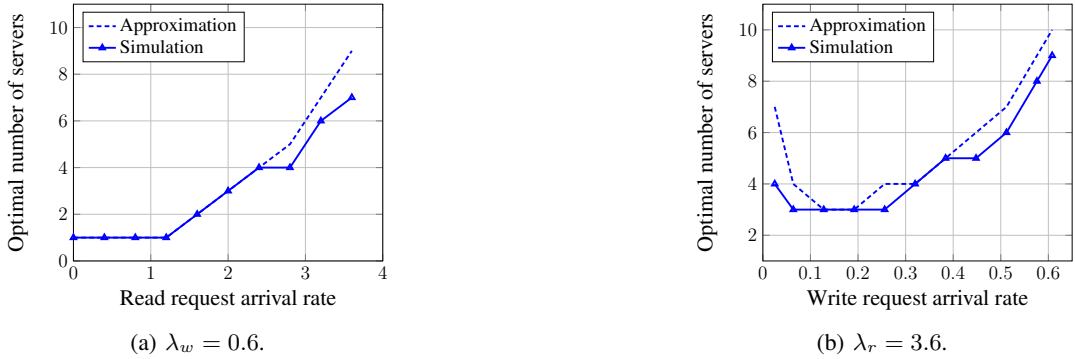
\begin{figure}[h]
\centering
\begin{subfigure}{.42\textwidth}
\centering
\begin{tikzpicture}[scale=0.6]
\pgfplotsset{every axis ylabel/.append style={font=\Large},
	xlabel/.append style={font=\Large}}
\pgfplotsset{every tick label/.append style={font=\large}}

\begin{axis}[
xlabel={Read request arrival rate},
ylabel={Optimal number of servers},
xmajorgrids,
ymajorgrids,
ymin=0,
ymax=11,
xmin=0,
xmax=4.0,
legend entries={Approximation, Simulation},
legend style = {legend pos = north west, nodes=right, font=\large},
]

\addplot[
color=blue,
line width=1.2pt,
densely dashed,
]
coordinates {
	(0.0,1) (0.4,1) (0.8,1) (1.2,1) (1.6,2) (2.0,3) (2.4,4) (2.8,5) (3.2,7) (3.6,9)
};

\addplot[
color=blue,
line width=1.2pt,
every mark/.append style={solid},
mark=triangle
]
coordinates {
	(0.0,1) (0.4,1) (0.8,1) (1.2,1) (1.6,2) (2.0,3) (2.4,4) (2.8,4) (3.2,6) (3.6,7)
};

\end{axis}
\end{tikzpicture}
\caption{$\lambda_w=0.6$.}
\label{fig:WPnOptVsReadAR}
\end{subfigure}%
\hspace{.075\textwidth}
\begin{subfigure}{.42\textwidth}
\centering
\begin{tikzpicture}[scale=0.6]
\pgfplotsset{every axis ylabel/.append style={font=\Large},
	xlabel/.append style={font=\Large}}
\pgfplotsset{every tick label/.append style={font=\large}}

\begin{axis}[
xlabel={Write request arrival rate},
ylabel={Optimal number of servers},
xmajorgrids,
ymajorgrids,
ymin=1,
ymax=11,
xmin=0,
xmax=0.65,
legend entries={Approximation, Simulation},
legend style = {legend pos = north west, nodes=right, font=\large},
]

\addplot[
color=blue,
line width=1.2pt,
densely dashed,
]
coordinates {
	(0.025,7) (0.064,4) (0.128,3) (0.192,3) (0.256,4) (0.32,4) (0.384,5) (0.448,6) (0.512,7) (0.576,9) (0.608,10)
};

\addplot[
color=blue,
line width=1.2pt,
every mark/.append style={solid},
mark=triangle
]
coordinates {
	(0.025,4) (0.064,3) (0.128,3) (0.192,3) (0.256,3) (0.32,4) (0.384,5) (0.448,5) (0.512,6) (0.576,8) (0.608,9)
};

\end{axis}
\end{tikzpicture}
\caption{$\lambda_r=3.6$.}
\label{fig:WPnOptVsWriteAR}
\end{subfigure}
\caption{The optimal number of servers as a function of request arrival rate in a write priority system. 
}
\label{fig:WPnOpt}
\end{figure}

We observe in Fig.~\ref{fig:WPnOptVsReadAR} that the optimal number of servers increase with read load while keeping the write load constant. 
This is due to the fact that write latency remains unaffected by the read load in write priority system, 
and decrease in read latency is faster than the increase in write latency due to number of servers. 
When we keep read load constant and increase the write load, 
write latency increases and read latency increases even faster. 
In this case, we observe an interesting phenomena in Fig.~\ref{fig:WPnOptVsWriteAR}, 
the optimal number of servers decreases until a write load threshold, and then increases. 

\subsection{Read priority system}
\label{subsection:RPNumericalStudies}

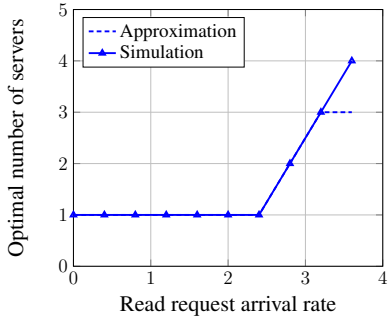
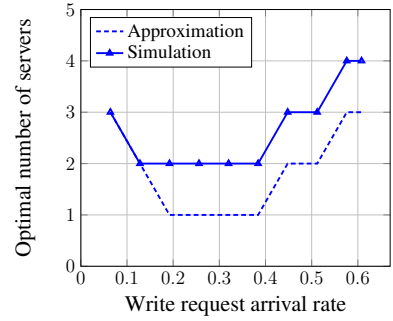
\begin{figure}[h]
\centering
\begin{subfigure}{.42\textwidth}
\centering
\begin{tikzpicture}[scale=0.6]
\pgfplotsset{every axis ylabel/.append style={font=\Large},
	xlabel/.append style={font=\Large}}
\pgfplotsset{every tick label/.append style={font=\large}}

\begin{axis}[
xlabel={Read request arrival rate},
ylabel={Optimal number of servers},
xmajorgrids,
ymajorgrids,
ymin=0,
ymax=5,
xmin=0,
xmax=4.0,
legend entries={Approximation, Simulation},
legend style = {legend pos = north west, nodes=right, font=\large},
]

\addplot[
color=blue,
line width=1.2pt,
densely dashed,
]
coordinates {
	(0.0,1) (0.4,1) (0.8,1) (1.2,1) (1.6,1) (2.0,1) (2.4,1) (2.8,2) (3.2,3) (3.6,3) 
};

\addplot[
color=blue,
line width=1.2pt,
every mark/.append style={solid},
mark=triangle
]
coordinates {
	(0.0,1) (0.4,1) (0.8,1) (1.2,1) (1.6,1) (2.0,1) (2.4,1) (2.8,2) (3.2,3) (3.6,4) 
};

\end{axis}
\end{tikzpicture}
\caption{$\lambda_w=0.6$.}
\label{fig:RPnOptVsReadAR}
\end{subfigure}%
\hspace{0.075\textwidth}
\begin{subfigure}{.42\textwidth}
\centering
\begin{tikzpicture}[scale=0.6]
\pgfplotsset{every axis ylabel/.append style={font=\Large},
	xlabel/.append style={font=\Large}}
\pgfplotsset{every tick label/.append style={font=\large}}

\begin{axis}[
xlabel={Write request arrival rate},
ylabel={Optimal number of servers},
xmajorgrids,
ymajorgrids,
ymin=0,
ymax=5,
xmin=0,
xmax=0.67,
legend entries={Approximation, Simulation},
legend style = {legend pos = north west, nodes=right, font=\large},
]

\addplot[
color=blue,
line width=1.2pt,
densely dashed,
]
coordinates {
	(0.064,3) (0.128,2) (0.192,1) (0.256,1) (0.32,1) (0.384,1) (0.448,2) (0.512,2) (0.576,3) (0.608,3)
};

\addplot[
color=blue,
line width=1.2pt,
every mark/.append style={solid},
mark=triangle
]
coordinates {
	(0.064,3) (0.128,2) (0.192,2) (0.256,2) (0.32,2) (0.384,2) (0.448,3) (0.512,3) (0.576,4) (0.608,4)
};

\end{axis}
\end{tikzpicture}
\caption{$\lambda_r=3.6$.}
\label{fig:RPnOptVsWriteAR}
\end{subfigure}
\caption{The optimal number of servers as a function of request arrival rate in a read priority system.}
\label{fig:RPnOpt}
\end{figure}

The behavior of optimal number of servers in read priority system remains similar to the write priority system. 
We observe in Fig.~\ref{fig:RPnOptVsReadAR} that the optimal number of servers increases with read load while keeping the write load constant. 
In this system, the write latency increases with read load. 
However, as the number of servers increases the dependence of write latency on read load reduces.  
Hence, the increase in number of servers decreases the read latency faster than increase in write latency. 
Whereas when the read load is fixed, 
we observe in Fig.~\ref{fig:RPnOptVsWriteAR} that the optimal number of servers decreases until a write load threshold, and then increases. 

It can be seen from Fig.~\ref{fig:WPnOpt} and~\ref{fig:RPnOpt} that the simulation results obtained from 
a system with practically applicable non-preemptive priority and round-robin routing of read requests,
remains close to the closed form approximation provided for simplified system with preemptive priority and random routing.

\subsection{Comparison of read and write priority}
In most practical distributed databases, 
the number of redundant servers $n$ is typically taken as two\footnote{\url{https://docs.mongodb.com/manual/core/replica-set-architectures/}}.
However, as we have observed for write priority system in Fig.~\ref{fig:WPnOptVsReadAR} and Fig.~\ref{fig:WPnOptVsWriteAR}, 
and for read priority system in Fig.~\ref{fig:RPnOptVsReadAR} and Fig.~\ref{fig:RPnOptVsWriteAR}, 
the optimal number of secondary servers $n$ can be different than two.  
The optimal redundancy can be different from two depending on the request arrival rate and service rate.
Hence, the typically chosen value of two redundant servers is not always optimal from latency perspective.
We have plotted the mean number of requests in the system for read and write priority systems with total number of servers in Fig.~\ref{fig:CompOptServ}, 
for two cases of read and write heavy systems. 
For the read heavy case, the read and write loads are fixed to be $\rho_r = 0.8$ and $\rho_w = 0.12$. 
For the write heavy case, the read and write loads are fixed to be $\rho_r = 0.12$ and $\rho_w = 0.8$.

\begin{figure}[h]
\centering
\begin{subfigure}{.42\textwidth}
\centering
\begin{tikzpicture}[scale=0.6]
\pgfplotsset{every axis ylabel/.append style={font=\Large},
	xlabel/.append style={font=\Large}}
\pgfplotsset{every tick label/.append style={font=\large}}

\begin{axis}[
name=WritePriority,
ylabel={ln(Mean number of requests)},
xlabel={Number of servers},
xmajorgrids,
ymin=3.55,
ymax=3.9,
xmin=1,
xmax=7,
colormap/blackwhite,
y tick label style={blue},
legend style = {legend pos = north east, nodes=right, font=\large},
]
\addplot[mark=*,blue, dashed, line width=1.2pt]
coordinates{
	(1,3.872016567063749) (2,3.601270852129078) (3,3.590551724395509) (4,3.5789588449372483) (5,3.5898357910785816) (6,3.6058265350150163)
};
\label{WriteHeavyWP}
\end{axis}

\begin{axis}[
name=ReadPriority,
axis y line*=right, axis x line=none,
ylabel={ln(Mean number of requests)},
xlabel={Number of servers},
xmajorgrids,
ymajorgrids,
ymin=2.39,
ymax=2.74,
xmin=1,
xmax=7,
colormap/blackwhite, 
y tick label style={black},
legend style = {legend pos = north east, nodes=right, font=\large},
]
\addlegendimage{/pgfplots/refstyle=WriteHeavyWP}\addlegendentry{Write Priority}
\addplot[mark=square,black, line width=1.2pt]
coordinates{
	(1,2.414010083206318) (2,2.523088822924899) (3,2.57031595610338) (4,2.6043446267362684) (5,2.637200912207317) (6,2.654450485393561)
};
\addlegendentry{Read Priority}
\end{axis}
\end{tikzpicture}
\caption{$\rho_w=0.8$ and $\rho_r=0.12$,}
\label{fig:CompLessOptServ}
\end{subfigure}%
\hspace{.075\textwidth}
\begin{subfigure}{.42\textwidth}
\centering
\begin{tikzpicture}[scale=0.6]
\pgfplotsset{every axis ylabel/.append style={font=\Large},
	xlabel/.append style={font=\Large}}
\pgfplotsset{every tick label/.append style={font=\large}}

\begin{axis}[
name=WritePriority,
ylabel={ln(Mean number of requests)},
xlabel={Number of servers},
xmajorgrids,
ymin=0.961,
ymax=1.031,
xmin=2,
xmax=17,
colormap/blackwhite,
y tick label style={blue},
legend style = {legend pos = north east, nodes=right, font=\large},
]
\addplot[mark=*,blue, dashed, line width=1.2pt]
coordinates{
	(4,1.0258811508642354) (5,0.9961786163256181) (6,0.9818042172933287) (7,0.974505292218307) (8,0.9709849022129217) (9,0.9696473826706538) (10,0.9684428318003947) (11,0.9685388223006646) (12,0.9691441838471195) (13,0.97052585864) (14,0.9723135935) (15,0.97443276965) (16,0.97517609727)
};
\label{ReadHeavyWP}
\end{axis}

\begin{axis}[
name=ReadPriority,
axis y line*=right, axis x line=none,
ylabel={ln(Mean number of requests)},
xlabel={Number of servers},
xmajorgrids,
ymajorgrids,
ymin=0.275,
ymax=0.345,
xmin=2,
xmax=17,
colormap/blackwhite, 
y tick label style={black},
legend style = {legend pos = north east, nodes=right, font=\large},
]
\addlegendimage{/pgfplots/refstyle=ReadHeavyWP}\addlegendentry{Write Priority}
\addplot[mark=square,black, line width=1.2pt]
coordinates{
	(3,0.33187815112953817) (4,0.2910443656909998) (5,0.28220652751813335) (6,0.28259860781972096) (7,0.29148131079440454) (8,0.2985448845238515) (9,0.30313201321406846)
};
\addlegendentry{Read Priority}
\end{axis}
\end{tikzpicture}
\caption{$\rho_w=0.12$ and $\rho_r=0.8$,}
\label{fig:CompMoreOptServ}
\end{subfigure}
\caption{The mean number of requests as a function of number of servers in the read and write priority system.}
\label{fig:CompOptServ}
\end{figure}
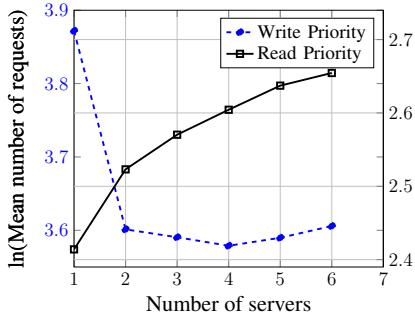
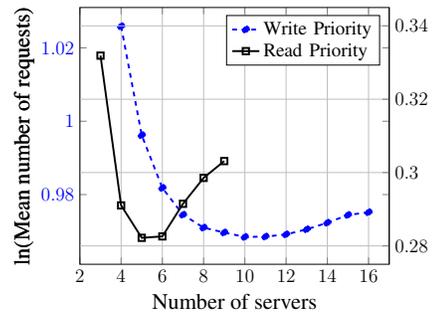

For write heavy system $\rho_w > \rho_r$, 
we observe that the optimal number of servers for read priority is unity. 
This follows from the fact that read load is already low, and the system latency is dominated by write load which gets worse with increasing number of redundant servers. 
However, for the same system with write priority, the optimal number of servers is larger than one. 
This is due to the following two facts. 
First, the read latency is affected by write load and resulting latency with write priority. 
Second, the increase in write latency due to increase in number of servers can be compensated by the decrease in read latency. 

For read heavy system $\rho_w < \rho_r$,  
we observe that the read priority system is highly sensitive to the choice of number of redundant servers.  
A small deviation from the optimal number can lead to significant changes in system latency. 
Whereas for write priority system in the same case, 
system latency is graceful to under or over provisioning of the number of redundant servers.

\subsection{Experimental Results}
\label{subsection:Exp}
\begin{figure}[h]
\centering
\begin{subfigure}{.42\textwidth}
\centering
\begin{tikzpicture}[scale=0.6]
\pgfplotsset{every axis ylabel/.append style={font=\Large},
    xlabel/.append style={font=\Large}}
\pgfplotsset{every tick label/.append style={font=\large}}

\begin{axis}[
xlabel={Number of servers},
ylabel={Mean number of requests},
xmajorgrids,
ymajorgrids,
ymin=0.95,
ymax=1.2,
xmin=1,
xmax=7,
legend entries={Experiment, Simulation},
legend style = {legend pos = north west, nodes=right, font=\large},
]

\addplot[
color=blue,
line width=1.2pt,
every mark/.append style={solid},
mark = *
]
coordinates {
     ( 1 , 1.1870236788325623 ) ( 2 , 1.1209315607241312 )  ( 3 , 1.0652315015350136 ) ( 4 , 1.0866816986704597 ) ( 5 , 1.0932372738297977 ) ( 6 , 1.0885733416664263 ) ( 7 , 1.0897919747772964 )   
};

%

\addplot[
color=blue,
line width=1.2pt,
densely dashed,
every mark/.append style={solid},
mark=triangle
]
coordinates {
     ( 1 , 0.9983396156163347 ) ( 2 , 0.9706662030910742 ) ( 3 , 0.9711213374417411 ) ( 4 , 0.9749929133668297 ) ( 5 , 0.9812034372246338 ) ( 6 , 0.9859738819667094 ) ( 7 , 0.9918474305088721 )

};

\end{axis}
\end{tikzpicture}
\caption{Write priority system with $\lambda_r=15$ and $\lambda_w=0.3$}
\label{fig:WPexp}
\end{subfigure}%
\hspace{.075\textwidth}
\begin{subfigure}{.42\textwidth}
\centering
\begin{tikzpicture}[scale=0.6]
\pgfplotsset{every axis ylabel/.append style={font=\Large},
    xlabel/.append style={font=\Large}}
\pgfplotsset{every tick label/.append style={font=\large}}

\begin{axis}[
xlabel={Number of servers},
ylabel={Mean number of requests},
xmajorgrids,
ymajorgrids,
ymin=0.53,
ymax=0.69,
xmin=1,
xmax=7,
legend entries={Experiment, Simulation},
legend style = {legend pos = north west, nodes=right, font=\large},
]

\addplot[
color=blue,
line width=1.2pt,
every mark/.append style={solid},
mark = *
]
coordinates {
 ( 1 , 0.711259575562535 )  ( 2 , 0.629920320197886 ) ( 3 , 0.6136389299009072 ) ( 4 , 0.5914708562722439 )  ( 5 , 0.5913639400378521 )  ( 6 , 0.5935076609087808 )     ( 7 , 0.608886177603269 )  
};


\addplot[
color=blue,
line width=1.2pt,
densely dashed,
every mark/.append style={solid},
mark=triangle
]
coordinates {
    
    ( 1 , 0.547461027758491 ) ( 2 , 0.5392430500505775 ) ( 3 , 0.5421576298009885 ) ( 4 , 0.5460796205513382 ) ( 5 , 0.5491643285306623 ) ( 6 , 0.5527105917386482 ) ( 7 , 0.5556737646514869 )

};

\end{axis}
\end{tikzpicture}
\caption{Read priority system with $\lambda_r=20$ and $\lambda_w=0.7$}
\label{fig:RPexp}
\end{subfigure}
\caption{
The mean number of requests as a function of number of servers in the read and write priority system obtained via experiment on $7$-server storage system.}
\label{fig:Exp}
\end{figure}
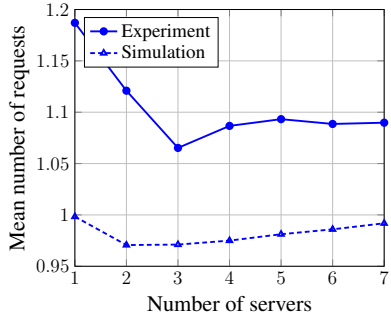
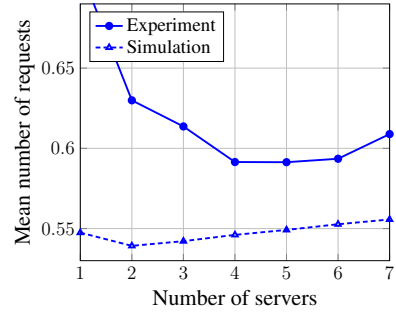

We perform experiments on a $7$-server distributed storage system with primary-secondary architecture and a single client server. 
Most database systems employ a random access memory (RAM) for serving client requests, 
while using a backend process to update the data into hard-disk periodically. 
This is due to the fact that read and write to cache is roughly $10$ times faster than read and write to storage disk.
Accordingly,  we read the file from the RAM of the storage system and wrote back to the RAM as well. 
While serving the read or write request to the specific file of interest, 
each of these servers had other background processes running (e.g. reads and writes to other files). 

The client server generates read and write request under Poisson arrival processes with rate $\lambda_r$ and $\lambda_w$ respectively. 
We empirically measure the service time distribution for read and write request to the RAM, 
by averaging it over all the seven servers. 
For read update for 150MB and write update of 200MB, 
we obtained the mean service time to read and write in RAM to be $0.021$s and $0.183$s respectively. 
We took the sum of number of requests in the system with service in progress or waiting for service, 
and obtained the mean number of requests by averaging them over time. 
We performed write priority experiment for a read arrival rate of $\lambda_r=15$ and write arrival rate of $\lambda_w=0.3$. 
We plot the mean number of requests with respect to the number of servers in Figure~\ref{fig:WPexp}. 
To compare, we simulated the system with the same arrival rate and empirically obtained service times distributions for read and write. 
We plot the mean number of requests in the read priority experiment and system simulation with identical distributions Figure~\ref{fig:RPexp}, 
with read arrival rate of $\lambda_r=20$ and write arrival rate of $\lambda_w=0.7$. 
Due to heterogeneity among servers, the read and write times at different servers do not have the same distribution. 
We observe in Figure~\ref{fig:Exp} that the latency curves obtained from experiment and the simulated system do not coincide, due to heterogeneity among servers and the additional system delays including network delays which were unaccounted in the simulations. 
However, even though the simulation and experiment do not match due to non-idealities in the system, 
we observe that there exists an optimal redundancy that minimizes the number of requests in the system.

\subsection{Read write systems with non-memoryless service distribution}
\begin{figure}[h]
\centering
\begin{subfigure}{.42\textwidth}
\centering
\begin{tikzpicture}[scale=0.6]
\pgfplotsset{every axis ylabel/.append style={font=\Large},
	xlabel/.append style={font=\Large}}
\pgfplotsset{every tick label/.append style={font=\large}}

\begin{axis}[
ylabel={Mean number of requests},
xlabel={Number of servers},
xmajorgrids,
ymajorgrids,
ymin=0.54,
ymax=0.59,
xmin=2,
xmax=7,
colormap/blackwhite,
legend entries={Empirical Distribution, Shited Exponential, Pareto, Weibull},
legend style = {legend pos = north west, nodes=right, font=\large},
]
\addplot[mark=*,blue, dashed, line width=1.2pt]
coordinates{
	(2,0.5547238832) (3,0.5466369748) (4,0.5457547086) (5,0.5462887209) (6,0.5481009661) (7,0.5495741394)
};

\addplot[every mark/.append style={solid},mark=square,black, loosely dashed, line width=1.2pt]
coordinates{
(2,0.5647617719) (3,0.5607229075) (4,0.5648201052) (5,0.56897213) (6,0.5720964689) (7,0.5758111264)
};

\addplot[every mark/.append style={solid}, mark=triangle,red, dotted, line width=1.2pt]
coordinates{
(2,0.5533997424) (3,0.5434513388) (4,0.5441000814) (5,0.5452300825) (6,0.5465479365) (7,0.5484581584)
};

\addplot[mark=diamond,green, solid, line width=1.2pt]
coordinates{
(2,0.5485615901) (3,0.543213552) (4,0.5448987879) (5,0.5470062578) (6,0.5490239504) (7,0.5510185706)
};

\end{axis}
\end{tikzpicture}
\caption{Write priority system with $\lambda_r=15$ and $\lambda_w=0.3$,}
\label{fig:WPdist}
\end{subfigure}%
\hspace{.075\textwidth}
\begin{subfigure}{.42\textwidth}
\centering
\begin{tikzpicture}[scale=0.6]
\pgfplotsset{every axis ylabel/.append style={font=\Large},
	xlabel/.append style={font=\Large}}
\pgfplotsset{every tick label/.append style={font=\large}}

\begin{axis}[
ylabel={Mean number of requests},
xlabel={Number of servers},
xmajorgrids,
ymajorgrids,
ymin=0.535,
ymax=0.585,
xmin=2,
xmax=7,
colormap/blackwhite,
legend entries={Empirical Distribution, Shited Exponential, Pareto, Weibull},
legend style = {legend pos = north west, nodes=right, font=\large},
]
\addplot[mark=*,blue, dashed, line width=1.2pt]
coordinates{
(2,0.5547238832) (3,0.5466369748) (4,0.5457547086) (5,0.5462887209) (6,0.5481009661) (7,0.5495741394)
};

\addplot[every mark/.append style={solid},mark=square,black, loosely dashed, line width=1.2pt]
coordinates{
	(2,0.5646819) (3,0.5588702977) (4,0.5611140271) (5,0.5641403598) (6,0.568321424) (7,0.5706532859)
};

\addplot[every mark/.append style={solid}, mark=triangle,red, dotted, line width=1.2pt]
coordinates{
	(2,0.5545863027) (3,0.5421814253) (4,0.5401967232) (5,0.5409136386) (6,0.5412098392) (7,0.5427523789)
};

\addplot[mark=diamond,green, solid, line width=1.2pt]
coordinates{
	(2,0.5498442519) (3,0.5430447961) (4,0.5441237012) (5,0.5448366871) (6,0.5467759938) (7,0.5484008682)
};

\end{axis}
\end{tikzpicture}
\caption{Read priority system with $\lambda_r=15$ and $\lambda_w=0.3$,}
\label{fig:RPdist}
\end{subfigure}
\caption{The mean number of requests as a function of number of servers in the read and write priority system under different distribution.}
\label{fig:DiffDist}
\end{figure}
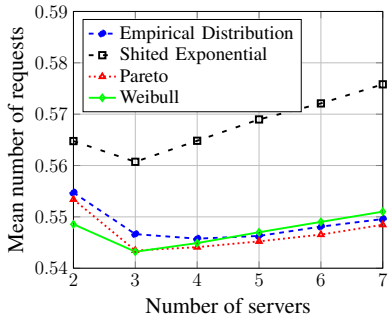
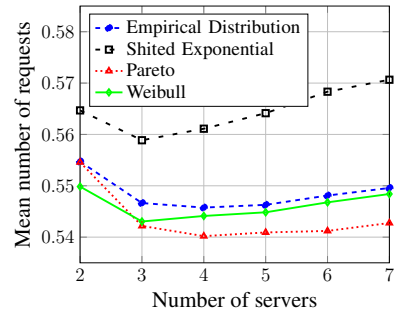

In this section, we study the effect of varying the service distribution on the optimal redundancy of the system.
To this end, ee plot the mean number of requests in a read-write system for the empirical distributions obtained from experiments in Section~\ref{subsection:Exp}, 
and for the closest fitting distributions from some parametrized families of non-memoryless distributions. 
Each of the read and write time distributions have a mean of $0.021$s and $0.183$s respectively.

The distribution function of a shifted exponential distribution with rate $\xi$ and shift $c$ can be represented as $F_{SE}(x) = (1 - e^{-\xi(x-c)})\SetIn{x\ge c}$.  
Fitting the empirical distribution with a shifted exponential distribution, 
we obtained the parameters $(\xi,c)=(136.096,0.015)$ for read service time and parameters $(\xi,c)=(12.43,0.105)$ for write service time.  
The distribution function for a pareto distribution with shape $\sigma$ and scale $x_m$ can be represented as
$F_{P}(x) = \Big(1 - \Big(\frac{x_m}{x}\Big)^\sigma\Big)\SetIn{x \ge x_m}$. 
Fitting the empirical distribution with a pareto distribution, 
we obtained the parameters $(\sigma,x_m)=(3.602,0.016)$ for read service time and parameters $(\sigma,x_m)=(5.203,0.15)$ for write service time.  
The distribution function for a weibull distribution with shape $k$, scale $\tau$ and location parameter $\theta$ can be represented as
$F_{W}(x) = 1 - e^{(\frac{x-\theta}{\tau})^k}\SetIn{x\ge \theta}$.
Fitting the empirical distribution with a weibull distribution, we obtain the parameters $(k,\tau,\theta)=(5.814,0.023,0)$ for read service time and parameters $(k,\tau,\theta)=(1.484,0.088,0.105)$ for write service time.
We simulate the system under different distributions for read arrival rate $\lambda_r = 15$ and write arrival rate $\lambda_w = 0.3$ under read and write priorities.
Plotting the mean number of requests with respect to the number of servers in Figure~\ref{fig:DiffDist},
we observe that there exists an optimal number of servers that minimises the mean number of requests under non-memoryless service as well.

\section{Conclusion and Future Directions}
We studied the latency-redundancy tradeoff in a distributed read-write system with Poisson arrivals and exponential service distribution. 
We provided novel closed-form approximations for mean number of write requests under read and write priorities. 
Under the proposed approximation, we characterized the optimal redundancy that minimizes the average request latency under read and write priorities.
We empirically showed that the optimal choice of redundancy under the proposed approximation closely follows the simulated result. 
We performed real world experiments and extensive numerical studies to demonstrate that the insights obtained from our theoretical study continue to hold true even in the real world settings and under non-memoryless service distributions. 

Our analysis framework can be extended to the study of multiple-file systems where each file is written to a subset of the servers. 
In such systems, read and write queues would themselves be multi-class queues where the request for different files can be considered as a separate class. 
Further, we are interested in finding the optimal redundancy for alternative system architectures and different consistency guarantees. 
Another interesting future direction would be to characterize the latency-redundancy tradeoff with general distribution for arrival and service processes.


\bibliographystyle{IEEEtran}
\bibliography{IEEEabrv,arXiv-2021}


\end{document}